\documentclass[12pt]{article}
\usepackage{authblk}

\usepackage{xcolor}
\usepackage[utf8]{inputenc}
\usepackage{hyperref}
\usepackage{enumitem}

\long\def\MSC#1\EndMSC{\def\arg{#1}\ifx\arg\empty\relax\else
     {\narrower\noindent%
{2010 Mathematics Subject Classification}: #1\\} \fi}
\long\def\PACS#1\EndPACS{\def\arg{#1}\ifx\arg\empty\relax\else
     {\narrower\noindent%
{PACS numbers}: #1}\fi}
\long\def\KEY#1\EndKEY{\def\arg{#1}\ifx\arg\empty\relax\else
	{\narrower\noindent%
Keywords: #1\\}\fi}

\usepackage{a4wide,amsmath}
\usepackage{amssymb}
\usepackage{mathtools}
\usepackage{mathrsfs}
\usepackage{amsthm, mathabx}
\numberwithin{equation}{section}

\theoremstyle{plain}
\newtheorem{theorem}{Theorem}[section]
\newtheorem{lemma}[theorem]{Lemma}
\newtheorem{proposition}[theorem]{Proposition}

\theoremstyle{definition}

\theoremstyle{remark}
\newtheorem{remark}[theorem]{Remark}
\setenumerate{label=(\roman*)}
\newcommand{\Z}{\mathbb{Z}}
\newcommand{\R}{\mathbb{R}}
\newcommand{\C}{\mathbb{C}}

\newcommand{\dist}{\mathrm{dist}}
\newcommand{\supp}{\mathrm{supp}}
\renewcommand{\Re}{\mathrm{Re}}
\renewcommand{\Im}{\mathrm{Im}}
\newcommand{\Tr}{\mathrm{Tr}}
\newcommand{\iu}{\mathrm{i}}
\newcommand{\e}{\mathrm{e}}

\newcommand{\K}{\mathfrak{K}}

\renewcommand{\d}{\;\mathrm{d}} 

\renewcommand{\epsilon}{\varepsilon}

\newcommand{\Mphase}{\phi}

\title{Bulk-edge correspondence  for unbounded Dirac-Landau operators}

\author[1]{H. D. Cornean}
\author[2]{M. Moscolari\footnote{Corr. author:  \href{mailto:massimo.moscolari@mnf.uni-tuebingen.de}{\texttt{massimo.moscolari@mnf.uni-tuebingen.de}}}}
\author[1]{K.S. S\o rensen}

\affil[1]{Department of Mathematical Sciences, Aalborg University, 
		 Skjernvej 4A, 9220 Aalborg, Denmark.} 
		\affil[2]{Fachbereich Mathematik, Eberhard-Karls-Universit\"at, 
		 Auf der Morgenstelle 10, 72076 T\"ubingen, Germany. 
		   }

{\footnotesize

\begin{document}

\maketitle

\begin{abstract}
We consider two-dimensional unbounded magnetic Dirac operators, either defined on the whole plane, or with infinite mass boundary conditions on a half-plane. Our main results use techniques from elliptic PDEs and integral operators, while their topological consequences are presented as corollaries of some more general identities involving magnetic derivatives of local traces of fast decaying functions of the bulk and edge operators. One of these corollaries leads to the so-called St{\v r}eda formula: if the bulk operator has an isolated compact spectral island, then the integrated density of states of the corresponding bulk spectral projection varies linearly with the magnetic field as long as the gaps between the spectral island and the rest of the spectrum are not closed, and the slope of this variation is given by the Chern character of the projection. The same bulk Chern character is related to the number of edge states which appear in the gaps of the bulk operator.

\end{abstract}


\section{Introduction, setting and main results}

The mathematics behind the zero temperature bulk-edge correspondence is  well understood in the context of discrete tight-binding models and continuous Schr\"odinger operators  \cite{KellendonkRichterSchulzBaldes, KellendonkSchulzBaldes, GrafPorta, ElgartGrafSchenker, CorneanMoscolariTeufel}. In particular, the analysis of tight-binding operators with  Dirac-like symmetries is also nowadays well developed \cite{PSB}.

In this work we are concerned with bulk-edge correspondence in the less explored setting of unbounded Dirac operators with constant magnetic fields, called from now on  Dirac-Landau operators. Our results and their proofs are entirely based on techniques coming from elliptic PDEs and integral operators of potential type (also known as polar operators). The equality between bulk and edge indices leading to the bulk-edge correspondence is shown as a corollary of a more general identity involving local traces of fast decaying functions of the bulk and edge operators, see Theorem \ref{thm:DiracBulkEdge}, in the same way as it is done in \cite{CorneanMoscolariTeufel}.

In this paper, the bulk operator is a magnetic Dirac operator acting in $L^2(\R^2)\otimes\C^2$, while the edge operator is constructed by restricting the bulk operator to the half-plane $E=\R\times [0,+\infty)$  with the so-called infinite mass boundary condition \cite{BS1, BS2}.

\medskip

In the rest of this section we describe the setting and formulate our main results, while in Section \ref{sec3} we provide a detailed construction of the Dirac-Landau operators with infinite mass boundary conditions, by specifying a domain of essential self-adjointness and analyzing how the resolvent operator behaves when the external magnetic field is varied. The corresponding main result is contained in Proposition \ref{prop:EssSelfAdj}. 

In Section \ref{sec4} we study the off-diagonal  localization of both the edge and bulk  resolvents as  functions of the spectral parameter, and also establish trace-class estimates for different functions of these operators. 

In Sections \ref{sec5} and \ref{sec6} we prove our two other main results, Theorem~\ref{thm:DiracBulkEdge} and Theorem~\ref{thm:DiracGapLabel}. Theorem~\ref{thm:DiracBulkEdge} extends the result contained in  \cite{CorneanMoscolariTeufel}, thus establishing a generalized bulk-edge correspondence for Dirac-Landau operators. The physical implications are briefly described in Section \ref{physics}. 

Theorem \ref{thm:DiracGapLabel} provides a relativistic version of the so-called St\v reda formula, or to be more precise it provides a gap-labelling theorem for relativistic operators (see the introduction of \cite{CorneanMonacoMoscolari} for a discussion about this nomenclature issue); however, here we follow the general folklore and call it simply St\v reda formula. The result of Theorem \ref{thm:DiracGapLabel} implies that if the bulk operator has an isolated spectral island, then the integrated density of states of the corresponding bulk spectral projection varies linearly with the magnetic field as long as the gaps between the spectral island and the rest of the spectrum are not closed. The slope of this variation is given by the Chern character of the projection. Moreover, in the spirit of bulk-edge correspondence, we show that the same bulk Chern character is related to the number of edge states  which appear in the old gaps of the bulk operator. This result generalizes the work of \cite{CorneanMonacoMoscolari} to Dirac-Landau operators.

\subsection{The setting}
We start with some preliminary results regarding the free Dirac operators. The bulk Hilbert space is $L^2(\mathbb{R}^2,\mathbb{C}^2) \equiv L^2(\mathbb{R}^2) \otimes \mathbb{C}^2$. We denote by  $\{\sigma_i\}_{i\in \{1,2,3\}}$ the Pauli matrices
\begin{align*}
\sigma_1= \begin{pmatrix}
0 & 1 \\
1 & 0
\end{pmatrix}, \quad
\sigma_2 = \begin{pmatrix}
0 & -\iu \\
\iu & \phantom{-}0
\end{pmatrix}, \quad
\sigma_3=\begin{pmatrix}
1 & \phantom{-}0 \\
0 & -1
\end{pmatrix},
\end{align*}
and $\sigma := (\sigma_1,\sigma_2)$. Let $H_0$ denote the bulk, zero mass free Dirac operator \cite{Th}:
\begin{align*}
H_0 = -\iu \nabla \cdot \sigma = p \cdot \sigma = -\iu \frac{\partial}{\partial x_1}\sigma_1 - \iu \frac{\partial}{\partial x_2} \sigma_2 = \begin{pmatrix}
0 & -\iu \partial_{x_1} - \partial_{x_2} \\
-\iu \partial_{x_1} + \partial_{x_2} & 0
\end{pmatrix}.
\end{align*}
 If $\mathscr{F}$ denotes the usual Fourier transform, we have
\begin{align*}
(\mathscr{F}H_0\mathscr{F}^{-1})(\xi) = 
\begin{pmatrix}
0 & \xi_1 - \iu \xi_2 \\
\xi_1 + \iu \xi_2 & 0
\end{pmatrix}.
\end{align*}

$H_0$ is elliptic, essentially self-adjoint on $C^\infty_0(\mathbb{R}^2)^2$, self-adjoint on $\mathscr{D}(H_0)=H^1(\mathbb{R}^2)^2$ and with purely absolutely continuous spectrum which covers the whole real axis \cite{Th}.

We want to model a constant magnetic field of strength $b \geq 0$ orthogonal to the plane. We introduce a magnetic vector potential in the Landau gauge  ${A}(x)=(-x_2,0)$ and we construct the bulk, purely magnetic, massless Dirac-Landau operator as
\begin{align*}
H_b &:= \big (p -b{A}(x) \big )\cdot \sigma , 
\end{align*}
which is also essentially selfadjoint on $C^\infty_0(\mathbb{R}^2)^2$. Even though $H_b$ can be written as $H_b=H_0+W_b$ where $W_b=-b{A}(x) \cdot \sigma$, the perturbation $W_b$ is not relatively bounded with respect to $H_0$ and the self-adjointness domain of $H_b$ varies with $b$, as we show in Proposition \ref{prop:EssSelfAdj}.

\medskip

The edge operator is given by the magnetic edge Dirac-Landau operator $H_b^E$ defined in $L^2(E,\mathbb{C}^2)$ for $E \coloneqq \{(x_1,x_2) \in \mathbb{R}^2 \mid  x_2 \geq 0\}$ with an infinite mass boundary condition \cite{BS1, BS2}. More precisely, let $\mathscr{S}_+$ be the set of functions which are restrictions to $E$ of Schwartz functions defined on  $\mathbb{R}^2$. Define the sets of functions
\begin{align}\label{hc6}
\mathscr{E}&:= \{ \psi=(\psi_1,\psi_2) \in \mathscr{S}_+ \oplus \mathscr{S}_+ \mid \psi_1(x_1,0)=\psi_2(x_1,0), \forall \, x_1 \in \mathbb{R}\},\nonumber \\
\mathscr{M} &:= \{ \psi \in \mathscr{E} \;| \;    \forall \, \alpha \in \mathbb{N}^2_0,\;  \exists \, c,C>0 \textup{ such that }\vert \partial^\alpha \psi \vert \leq C e^{-c \vert x \vert} \}.
\end{align}

Denote by $\tilde{H}_b$ the symmetric operator which acts as $(-i \nabla -b{A})\cdot\sigma$ on functions from $\mathscr{M}$. Then, the purely magnetic edge Dirac-Landau  operator is defined as the closure of $\tilde{H}_b$:
\begin{equation}\label{hc1}
     H_b^E:= \overline{\tilde{H}_b},\quad b\geq 0.
\end{equation}

\subsection{Main results} We start with a technical proposition which deals with the resolvent of the edge operator:
\begin{proposition}\label{prop:EssSelfAdj} Let $\Mphase(x,x') := (x'_1-x_1)x'_2$.
The following results hold true:
\begin{enumerate}[label=(\roman*)]
    \item \label{thm:1Part1} Let $\lambda>0$. The operator $H^E_b$ restricted to $\mathscr{M}$ is essentially self-adjoint. Moreover, the domain $D(H^E_b)$ coincides with the range of the operator $S_b(\iu\sqrt{\lambda})$, which has an explicit integral kernel defined by \eqref{eq:Slambda}.
    \item \label{thm:1Part2} Let $b_0>0$.  There exists $\lambda_0>0$ such that for all $\lambda>\lambda_0$ there exists an integral operator family $K_b(\iu\sqrt{\lambda})$, which is smooth with respect to $b\in (0,b_0)$ in the operator norm topology, and by denoting with $(H^E_b-\iu\sqrt{\lambda})^{-1}(x,y)$ the integral kernel of the resolvent we have:
    \begin{equation}\label{eq:HEexpansion}
        (H^E_b-\iu\sqrt{\lambda})^{-1}(x,y)=\e^{\iu b\Mphase(x,y)}K_b(\iu\sqrt{\lambda})(x,y). 
    \end{equation}
\end{enumerate}
\end{proposition}

\vspace{0.2cm}

\begin{remark}\label{remarkhc1}
All our next results also hold true when one adds a bounded $\Z^2$-periodic perturbation to $H_b$. A trivial example of such a perturbation is a mass term $m\sigma_3$ with $m>0$, not necessarily constant, see \cite{MoscolariStottrup} for similar results in the Schr\"odinger setting. 

One may even add an Anderson-like random potential, at the price of an averaging over the random variables; we chose though not to work with random perturbations in order to simplify the presentation and highlight the real technical difficulty, which is the presence of a constant magnetic field combined with the infinite mass boundary condition.
\end{remark}

The next theorem is an extension to Dirac operators of the main result of \cite{CorneanMoscolariTeufel}:

\begin{theorem}
\label{thm:DiracBulkEdge}
Let $\Omega = [0,1]^2$,  $S_L\coloneqq [0,1]\times [0,L]$ for $L\geq 1$, and let $\chi_\Omega$, $\chi_L$  be the indicator functions of their respective sets.  Let $\chi_\infty$ denote the indicator function of the semi-infinite strip with  $L=\infty$. Also, let $f$ be a real valued Schwartz function on $\mathbb{R}$. Then:
\begin{enumerate}
    \item \label{thm:DiracBulkEdge1} Both operators $\chi_\Omega f(H_b)$ and $\chi_Lf(H_b^E)$ are  trace class. Let 
    \begin{align}\label{hc3}
        \rho_L(b) \coloneqq \frac{1}{L}\Tr (\chi_Lf(H_b^E)), \qquad B_f(b) \coloneqq \Tr (\chi_\Omega f(H_b)).
    \end{align}
    Then both $\rho_L$ and $B_f$ are differentiable as functions of the magnetic field $b$ and:
    \begin{align}\label{hc3'}
        \lim_{L\to \infty}\rho_L(b)=B_f(b)\quad { and}\quad \lim_{L \to \infty} \frac{\d \rho_L}{\d b}(b) = \frac{\d B_f}{\d b}(b).
    \end{align}
    \item \label{thm:DiracBulkEdge2} Let $g \in C^1([0,1])$ with $g(0)=1$ and $g(1)=0$. By setting $\tilde{\chi}_L(x) \coloneqq \chi_L(x)g(x_2/L)$ we have: 
    \begin{align*}
        \frac{\d B_f}{\d b}(b) = - \lim_{L\to \infty}  \Tr\big (\tilde{\chi}_L  \iu [H_b^E,X_1]\, f'(H_b^E)\big ),\quad \iu [H_b^E,X_1]=\sigma_1.
    \end{align*}
    
    \item \label{thm:DiracBulkEdge3} If $f'$ is supported in a spectral gap of $H_b$, then ${\chi}_\infty f'(H_b^E)$ is  a trace class operator and 
    \begin{align}\label{hc2}
        \frac{\d B_f}{\d b}(b) = - \Tr\big ({\chi}_\infty  \sigma_1 f'(H_b^E)\big ).
    \end{align}
    \item \label{thm:DiracBulkEdge4} If $H_b$ is the purely magnetic massless Dirac-Landau operator, and if the support of $f'$ is disjoint from $\sigma(H_b)=\{\pm \sqrt{2nb}\; |\; n\geq 0\}$, then:
    \begin{align}\label{hc2'}
        - \Tr\big ({\chi}_\infty  \sigma_1 f'(H_b^E)\big )=\frac{1}{2\pi}\sum_{k\in\Z} f\big ( {\rm sgn}(k)\sqrt{2|k|b}\big ).
    \end{align}
    
\end{enumerate}
\end{theorem}

\vspace{0.2cm}
\begin{remark}
Formula \eqref{hc2'} can be directly obtained by combining \eqref{hc2} with another identity which holds for all real valued Schwartz functions $f$: 
\begin{align}\label{dc2}
B_f(b)=\frac{b}{2\pi}\sum_{k\in\Z} f\big ( {\rm sgn}(k)\sqrt{2|k|b}\big ).
\end{align}
 A few hints regarding the proof of \eqref{dc2} can be found at the end of Section \ref{sec5}. We choose though to also directly prove \eqref{hc2'} in Section \ref{sec5}, as a nice illustration of the bulk-edge correspondence. 
\end{remark}

Now we switch to a slightly different topic. In the case of the bulk Dirac-Landau operator one can explicitly see that the gaps vary continuously with $b$, but this is a more general phenomenon as we will see in Section \ref{sec6}, see also  \cite{BC} for a proof of gap stability with respect to long-range magnetic perturbations not necessarily coming from constant magnetic fields. Thus we can define the Riesz spectral projections corresponding to a compact spectral island $\sigma_0(b)$ contained in the spectrum of $H_b$: 
\begin{align*}
    \Pi_b=\frac{\iu }{2\pi} \oint_{\mathcal{C}} (H_b -z)^{-1}\d z.
\end{align*}
Here $\mathcal{C}$ is a positively oriented simple contour that contains $\sigma_0(b)$ and which stays at a positive distance to the spectrum as $b$ varies in some narrow compact interval.  We show in Section \ref{sec6} that the projection $\Pi_b$ has a  jointly continuous integral kernel, with exponential decay away from the diagonal, and by multiplying the projection with the indicator function of a bounded set we obtain a  trace class operator. 

Recall the notation $\Mphase(x,\eta)=(\eta_1-x_1)\eta_2$. The magnetic translations defined by  $(\tau_{b,\eta}\psi)(x):=e^{\iu b \Mphase(x,\eta)}\psi(x-\eta) $ commute with $H_b$ due to the identity 
\begin{equation}\label{eq:HbEalmostcommutingwithmagneticphase}
\big (-\iu\nabla_x-b{A}(x)\big )e^{\iu b\Mphase(x,\eta)}=e^{\iu b\Mphase(x,\eta)}\big (-\iu\nabla_x-b{A}(x-\eta)\big ).
\end{equation}
Hence they also commute with the resolvent and thus with the Riesz projections, which leads to
\begin{align*}
    \Pi_b(x,x')=\e^{\iu b\Mphase(x,\eta)} \Pi_b(x-\eta,x'-\eta) \e^{-\iu b\Mphase(x',\eta)}.
\end{align*}
In particular, this shows that for the purely magnetic Dirac-Landau operator, the diagonal value  $\Pi_b(x,x)$ is a $2\times 2$ matrix independent of $x$. In fact, one can show that by taking the $\C^2$ trace of $\Pi_b(x,x)$ in this case, one gets a contribution of $b/(2\pi)$ from each Dirac-Landau level. When $H_b$ is perturbed by a $\Z^2$-periodic potential and/or mass term which does not close the gaps between the spectral island $\sigma_0(b)$ and the rest of the spectrum, the diagonal value  $\Pi_b(x,x)$ is a continuous $\Z^2$-periodic matrix valued function. 

We can now define the integrated density of states for $\Pi_b$, as 
\begin{align*}
    \mathcal{I}(\Pi_b) \coloneqq \lim_{L\to\infty} \frac{1}{|S_L|}{\rm Tr} (\chi_L \Pi_b)=\int_\Omega {\rm tr}_{\C^2}\Pi_b({x},{x}) \d {x},
\end{align*}
where the second equality follows from the periodicity of the diagonal value. 

The next result is an extension to Dirac operators of the gap labelling theorem shown in \cite{CorneanMonacoMoscolari}. We formulate it so that it covers the case when $\Z^2$-periodic perturbations are present.
\begin{theorem}
\label{thm:DiracGapLabel}
Let us assume that $H_b$ has a compact spectral island $\sigma_0(b)$ which varies  continuously (as a set) in the Hausdorff distance for $b \in (b_1,b_2)$. Let $\Pi_b$ be the Riesz projection associated with $\sigma_0(b)$. Then the map $b \mapsto \mathcal{I}(\Pi_b)$ is continuously differentiable for $b \in (b_1,b_2)$, and the corresponding spectral projection $\Pi_b$ satisfies
\begin{align}\label{hc4}
    \frac{\d \mathcal{I} (\Pi_b)}{\d b} = \frac{1}{2\pi}\mathrm{Ch}(\Pi_b),
\end{align}
where 
\begin{align*}
    \mathrm{Ch}(\Pi_b)\coloneqq 2\pi \int_\Omega {\rm tr}_{\C^2}\big (i \Pi_b [[X_1,\Pi_b],[X_2,\Pi_b]]\big )(x,x) \d x \in \Z
\end{align*}
is an integer valued constant.
\end{theorem}

\vspace{0.2cm}

\begin{remark}\label{remarkhc2}
Let us assume that $0\leq f\leq 1$ is a smooth compactly supported function such that $f=1$ on $\cup_{b \in (b_1,b_2)}\sigma_0(b)$ and $f'$ is supported in $\R \setminus \cup_{b \in (b_1,b_2)}\sigma_0(b)$. From the spectral theorem we have $\Pi_b=f(H_b)$, and by comparing with \eqref{hc3} we also have $\mathcal{I}(\Pi_b)=B_f(b)$. Thus, using \eqref{hc4} in \eqref{hc2} we get a zero-temperature bulk-edge correspondence of the form
\begin{equation}
\label{eq:ZeroTBulk-edge}
- \Tr\big ({\chi}_\infty  \sigma_1 f'(H_b^E)\big )=\frac{1}{2\pi}\mathrm{Ch}(\Pi_b) \, .
\end{equation} 
A consequence of \eqref{eq:ZeroTBulk-edge} is that  if $\mathrm{Ch}(\Pi_b)\neq 0$, then $f'(H_b^E)$ cannot be identically zero, hence the spectrum of the edge Hamiltonian must completely cover at least one of the two bulk gaps defining $\sigma_0(b)$. There are though situations when a piece of one of the two bulk gaps can survive, as is the case for the  purely magnetic massless edge operator: the gap $(-\sqrt{2b},0)$ never closes completely, only the others do, see Theorem 4.3 in \cite{BS1}. 
\end{remark}

\subsection{Positive temperature bulk-edge correspondence}\label{physics}
When the bulk Dirac operator and the edge Dirac operator have a joint spectral gap that contains zero, Theorem \ref{thm:DiracBulkEdge} allows to extend the general non-relativistic bulk-edge correspondence of  \cite{CorneanMoscolariTeufel} to the setting of a relativistic Fermi gas containing both electrons and positrons. Indeed, by setting the function $f$ of Theorem \ref{thm:DiracBulkEdge} to be equal to $f^{(e_-)}(x):=-T\chi_{[0,+\infty)}(x) \ln \big (1+ \e^{-(x-\mu)/T}\big )$ we get that
$$
p^{e_-}(b,T,\mu):=-B_{f^{(e_-)}}(b)$$ is the grand canonical pressure of a relativistic gas of electron at temperature $T>0$ and chemical potential $\mu >0 $ \cite{LandauLifschitz,MillerRay1984}. Then, since the positrons are interpreted as holes in the Dirac sea we have that the chemical potential of the positron gas is $-\mu$ \cite{LandauLifschitz,MillerRay1984}, and the positron grand canonical pressure is given by $p^{e_+}(b,T,-\mu):=-B_{f^{(e_+)}}$, where $f^{(e_+)}(x):=-T\chi_{(-\infty,0]}(x) \ln \big (1+ \e^{-(-x+\mu)/T}\big )$. Notice that 
\begin{align*}
\mp \partial_\mu p^{e_\pm}(b,T,\mp \mu)=n^{e_\pm}(b,T,\mp \mu)
\end{align*} where $n^{e_-}(b,T,\mu)=B_{F_{FD}^{(e_-)}}$ and $n^{e_+}(b,T,-\mu)=B_{F_{FD}^{(e_+)}}$ are the electrons and positrons density respectively, which can be rewritten as  
\begin{align*}
n^{e_\pm}(b,T,\mp \mu)&=B_{F_{FD}^{(e_\pm)}}
\end{align*} with $F_{FD}^{(e_-)}:=\chi_{[0,+\infty)}(x)(1+\e^{(x-\mu)/T})^{-1}$, and $F_{FD}^{e_+}:=\chi_{(-\infty,0]}(x)(1-(1+\e^{(x-\mu)/T})^{-1})$.

Then, the total grand canonical magnetization $m(b,T,\mu)$ is simply given by the magnetic derivative of the total pressure $p(b,T,\mu)=p^{e_+}(b,T,-\mu)+p^{e_-}(b,T,\mu)$, therefore Theorem~\ref{thm:DiracBulkEdge} shows that
$$
m(b,T,\mu)=I(b,T,\mu)
$$
where $I(b,T,\mu)$ is the total edge current. Notice that $I(b,T,\mu)$ has two components, one which describes an electron edge current $I^{e_-}(b,T,\mu)$ and another one that describes a positron edge current $I^{e_+}(b,T,-\mu)$, that is 
\begin{equation}
\label{eq:MequalsIDirac}
I^{e_\pm}(b,T,\mp\mu)=\mp \lim_{L\to \infty}  \Tr\big (\tilde{\chi}_L  \iu [H_b^E,X_1]\,  F_{FD}^{(e_\pm)}\left(H_b^E\right)\big ).
\end{equation}
The equation \eqref{eq:MequalsIDirac} extends the general bulk-edge correspondence shown in \cite{CorneanMoscolariTeufel} to our relativistic setting. Eventually, the relation \eqref{eq:MequalsIDirac} can be connected to the zero-temperature bulk-edge correspondence \eqref{eq:ZeroTBulk-edge} by taking the derivative with respect to $\mu$ and performing a zero-temperature limit as in \cite[Proposition 1.5]{CorneanMoscolariTeufel}.

\subsection{Comparison with previous results}
\label{literature}
The mathematical analysis of bulk-edge correspondence has received a lot of attention during the last two decades.  Most results have dealt with  tight-binding models and continuous Schr\"odinger operators. We refer the interested reader to the introduction of \cite{CorneanMoscolariTeufel} for a recent account on the history of the problem.

Nevertheless, much less is known for the case of continuous Dirac operators. In \cite{Bal}, the author extends the pair of projection formalism \cite{AvronSeilerSimon} to the setting of free Dirac operators and proves a bulk-edge correspondence in the junction framework, i.e. when the edge is created by considering a position dependent mass term for the Dirac operator. Since the massive free Dirac operator has two unbounded spectral islands, in \cite{Bal} the bulk operator has to be regularized with a second order differential term. Moreover, the properties of edge modes in the junction setting have been thoroughly analyzed \cite{BBDFKLW,BalBeckerDrouot,Drouot} by means of semiclassical techniques. Regarding the study of bulk-edge correspondence where the edge operator is defined on a half-plane, we are only aware of the work \cite{TauberetAl} where the authors show a violation of the bulk-edge correspondence by imposing Dirichlet boundary conditions to the bulk free Dirac operator regularized with a second order term as in \cite{Bal}. A similar mechanism has been analyzed in \cite{GrafJudTauber} for the shallow water model.

\subsection{Open questions}
\begin{enumerate}
    \item Another physically important boundary condition for Dirac operators is the so-called ``zig-zag"  \cite{BenguriaFournaisStockmeyerVanDenBosch}. The results of our paper cannot be applied to such a setting. A more delicate analysis is required and we postpone it to a future work.
    
    \item An alternative way to compute the integrated density of states for magnetic operators, by using the Dixmier trace, has been recently developed in a series of works  \cite{DeNittisGomiMoscolari,DeNittisSandoval,DeNittisBelmonte}. We expect to  be possible to extend such results in order to obtain and generalize formulas like in \eqref{dc2}. 
    
    \item In \cite{RaghuHaldane} Raghu and Haldane showed that photonic crystals exhibit a bulk-edge correspondence like their electronic counterparts. The mathematical properties of bulk models for photonic crystals have been thoroughly investigated in the last years (see \cite{DeNittisLein} and references therein) by using Maxwell operators, which are first-order differential operators that have many similarities with the bulk Dirac operators considered here. An interesting problem would be to characterize edge Maxwell operators similarly to what we do in Proposition  \ref{prop:EssSelfAdj} and to suitably extend Theorem~\ref{thm:DiracBulkEdge} in order to enlighten the mathematics behind bulk-edge correspondence in the framework of Maxwell operators \cite[Section IV]{DeNittisLein}.
\end{enumerate}

\addtocontents{toc}{\protect\setcounter{tocdepth}{0}}

\addtocontents{toc}{\protect\setcounter{tocdepth}{2}}

\section{Dirac-Landau operators with infinite mass boundary conditions}
\label{sec3}

\subsection{The integral kernel of the free bulk Dirac operator}
Consider the free, massless  bulk Dirac operator $H_0=p\cdot\sigma$. Let $\lambda>0$. An elementary computation gives
\begin{align*}
(H_0+i\sqrt{\lambda})(H_0-\iu\sqrt{\lambda})=(H_0^2+\lambda)=(p^2+\lambda)I_2,
\end{align*}
which shows that 
\begin{align*}
(H_0-\iu\sqrt{\lambda})^{-1}=(H_0+\iu \sqrt{\lambda})(p^2+\lambda)^{-1}I_2= \iu \sqrt{\lambda}(p^2+\lambda)^{-1}I_2+\sigma\cdot p(p^2+\lambda)^{-1} .
\end{align*}
 Thus the resolvent is an integral operator with a $2\times 2$ matrix-valued integral kernel given by   
\begin{align}
\label{eq:ResolventIntegral}
 \K_0(x,x';\sqrt{\lambda})=\frac{1}{(2\pi)^2}\int_{\R^2} \e^{\iu \xi\cdot (x-x')} \Big ( \iu \sqrt{\lambda}(\xi^2+\lambda)^{-1}I_2+\sigma\cdot \xi (\xi^2+\lambda)^{-1}\Big )\; \d \xi.
\end{align}
These integrals can be expressed in terms of the  Macdonald function $K_0$ and its derivative (see \cite{JN, NU}). First we have:
$$
K_0(\sqrt{\lambda}\vert x\vert)=(2\pi)^{-1}\int_{\R^2} \e^{\iu \xi\cdot x}\frac{1}{\xi^2 +\lambda}\; \d \xi,\quad -\iu \sqrt{\lambda}\frac{x}{|x|}K_0'(\sqrt{\lambda}|x|)=
(2\pi)^{-1}\int_{\R^2} \e^{\iu \xi\cdot x}\frac{\xi}{\xi^2 +\lambda}\; \d \xi.$$
 Let us briefly recall some of the properties of $K_0$ \cite{NU}. The Macdonald functions are smooth and positive functions on $(0,\infty)$ and they satisfy the following properties
\begin{enumerate}
\item $\max \{K_0(t), |K_0'(t)|\}\leq  C\; \e^{-t}$ for $t \geq 1$,
\item $K_0(t) \sim - \log(t)$ and $|K_0'(t)| \sim t^{-1}$ for $t \to 0$.
\end{enumerate}

Also, there exist some $C,c>0$ such that
\begin{align}
\vert K_0( \vert x \vert )\vert \leq C \big (1+ \vert \log ( \vert x \vert)\vert \big ) \e^{-c\vert x \vert}, \quad
\vert K_0'( \vert x \vert )\vert \leq C \big(1+\vert x \vert^{-1}\big) \e^{-c\vert x \vert},\quad \forall \, x\neq 0. \label{eq:K1estimate}
\end{align}

Then the integral kernel of the resolvent of the free massless bulk Dirac operator can be explicitly written as
\begin{align}\label{eq:integralkernelK}
\K_0(x,x';\sqrt{\lambda})=(2\pi)^{-1} \iu \sqrt{\lambda}K_0(\sqrt{\lambda}|x-x'|)I_2 -\iu \sqrt{\lambda} (2\pi)^{-1}\sigma \cdot \frac{x-x'}{|x-x'|} K_0'(\sqrt{\lambda}|x-x'|).
\end{align}
Finally, we remark that we have the following distributional identity
\begin{align}\label{hc10}
 (-\iu \nabla_x \cdot \sigma -\iu \sqrt{\lambda})\K_0(x,x';\sqrt{\lambda})=\delta(x-x')\; I_2.
\end{align}

\subsection{Properties of the free edge Dirac operator}

Consider the extension operator defined by
\begin{align}\label{eq:phitwidle}
J:L^2(E,\mathbb{C}^2) \to  L^2(\mathbb{R}^2,\mathbb{C}^2),\quad (J \varphi) (x_1,x_2):=\begin{cases}
\varphi(x_1,x_2) & \textup{if } x_2 \geq 0 \\
(\sigma_1\varphi)(x_1,-x_2) & \textup{if } x_2 < 0
\end{cases}\; .
\end{align}
Define the unitary operator $U \colon L^2(\mathbb{R}^2,\mathbb{C}^2) \to L^2(\mathbb{R}^2,\mathbb{C}^2)$ by
\begin{align}\label{hc7}
(Uf)(x_1,x_2)\coloneqq (\sigma_1f)(x_1,-x_2).
\end{align}
An explicit computation shows that $U J \varphi = J \varphi$, for every $\varphi\in C_0^\infty(E,\C^2)$. This shows that $J\varphi|_{x_2\geq 0}\in \mathscr{M}$ (see \eqref{hc6}). 

\begin{lemma}\label{lemmahc1}
For every $\varphi \in C^\infty_0(E,\mathbb{C}^2)$ it holds true that
\begin{align}\label{hc9}
\Big ((H_0-\iu \sqrt{\lambda})^{-1}J\varphi\Big )|_{x_2\geq 0}\in \mathscr{M}.
\end{align}
\end{lemma}
\begin{proof}
By considering the explicit expression \eqref{eq:ResolventIntegral} of the integral kernel of $H_0$ we have
\begin{align}\label{eq:integralkernelresolventbulk_nomagne}
&\Big[(H_0-i\sqrt{\lambda})^{-1}J{\varphi}\Big](x_1,x_2) =\int_\mathbb{R} \d x_1'\Big (\int_{-\infty}^0\d x_2' +\int_0^\infty \d x_2'\Big )  \K_0(x,x';\sqrt{\lambda})(J{\Psi})(x') \\
&= \int_\mathbb{R}\d x_1' \int_0^\infty \d x_2' \bigg(  \K_0(x, (x_1',-x_2');\sqrt{\lambda})(J{\varphi})(x_1',-x_2') + \K_0(x,x';\sqrt{\lambda})(J{\varphi})(x_1',x_2') \bigg)  \nonumber \\ 
&= \int_\mathbb{R}\d x_1' \int_0^\infty \d x_2'   \K_0(x, (x_1',-x_2');\sqrt{\lambda})\sigma_1 \varphi(x_1',x_2') + \int_\mathbb{R}\d x_1' \int_0^\infty \d x_2' \K_0(x,x';\sqrt{\lambda}) \varphi(x_1',x_2'). \nonumber 
\end{align}
Because $\varphi$ equals zero when $x_2'$ is small enough, the first integral can be restricted to an interval of the type $[\epsilon,\infty)$ with $\epsilon >0$. Since $x_2\geq 0$, we have $x_2-(-x_2')\geq \epsilon>0$, hence the kernel in the first integral above does not see the singularity on its diagonal, thus this integral generates a smooth function on $E$.  The second integral can be extended to $x_2'\in \R$ and (up to a numerical constant) it equals the inverse Fourier transform of $(\xi\cdot \sigma -\iu \sqrt{\lambda})^{-1}(\mathscr{F}\varphi)(\xi) $, which is a  Schwartz function, thus smooth. For similar  regularity issues related to Schr\"odinger operators, see Lemma A.1 in  \cite{MoscolariStottrup}.

Another explicit computation shows that $[H_0,U]=0$, hence $U$ also commutes with the resolvent of $H_0$. Together with $UJ\varphi=J\varphi$, this implies that  $U(H_0-\iu \sqrt{\lambda})^{-1}J{\varphi} =  (H_0-\iu \sqrt{\lambda})^{-1}J{\varphi}$ which shows that the boundary condition at $x_2=0$ is satisfied. Finally, the exponential decay is insured by the compact support of $\varphi$ and the exponential decay of the Macdonald function $K_0$ and its derivatives. 
\end{proof}

We are now ready to complete the proof of Proposition \ref{prop:EssSelfAdj} when $b=0$. We do this by showing that the range of $\tilde{H}_0\pm \iu \sqrt{\lambda}$ is dense in $L^2(E,\C^2)$ (see e.g. \cite[Theorem X.1]{ReedSimon2}). Using \eqref{hc10} and \eqref{hc9} we have 
$$(\tilde{H}_0-\iu \sqrt{\lambda})\big [(H_0-\iu \sqrt{\lambda})^{-1}J{\varphi}\big ]|_{x_2\geq 0}=\varphi,\quad \forall \, \varphi\in C_0^\infty (E,\C^2).$$
A similar argument works for $+\iu \sqrt{\lambda}$, and this shows that $\tilde{H}_0$ is essentially self-adjoint and its closure equals $H_0^E$. Moreover, this proof also provides us with an integral kernel for the edge resolvent, which can be read off from \eqref{eq:integralkernelK}: \begin{align}
 \K_0^E(x,x';\sqrt{\lambda}):=\K_0(x,  (x_1',-x_2');\sqrt{\lambda})\sigma_1 + \K_0(x,x';\sqrt{\lambda}) . \label{hc11}
\end{align}

\subsection{Proof of Proposition \ref{prop:EssSelfAdj}(i): the magnetic edge operator}
We recall the formula of the non-symmetric magnetic phase 
\begin{align*}
    \Mphase (x,x')=(x_1'-x_1)x_2',\quad \forall \, x,x' \in \mathbb{R}^2.
\end{align*}

Let $S_b(\iu \sqrt{\lambda})$ and $T_b(\iu \sqrt{\lambda})$ be the operators defined on $L^2(E,\C^2)$ which have integral kernels
\begin{align}\label{eq:Slambda}
    S_b(x,x';\sqrt{\lambda}) &\coloneqq \e^{\iu b\Mphase(x,x')}\K_0^E(x,x';\sqrt{\lambda}) \\
    T_b(x,x';\sqrt{\lambda}) &\coloneqq \e^{\iu b\Mphase(x,x')}(-b{A}(x-x')\cdot \sigma)\K_0^E(x,x';\sqrt{\lambda}).\nonumber 
\end{align}

Before showing the selfadjointness of $H_b^E$, we need to prove two preliminary technical lemmas regarding the properties of $S_b(\iu \sqrt{\lambda})$ and $T_b(\iu \sqrt{\lambda})$.

\begin{lemma}\label{prop:boundednessSbTb}
The operators $S_b$ and $T_b$ are bounded on $L^2(E,\mathbb{C}^2)$. Moreover, there exists $C>0$ such that for all $b, \lambda>0$ we have $\Vert T_b(\iu \sqrt{\lambda}) \Vert \leq C\; b\; \lambda^{-1}$.
\end{lemma}
\begin{proof}
From \eqref{hc11},  \eqref{eq:integralkernelK} and \eqref{eq:K1estimate} we conclude that there exists a function $F$ such that $\e^{\alpha |\cdot|}F(|\cdot|)\in L^1(\R^2)$ for some $\alpha>0$ and  $$\Vert\K_0^E(x,x';\sqrt{\lambda})\Vert_{\C^2}\leq \sqrt{\lambda}\;  F(\sqrt{\lambda}|x-x'|),\quad \forall \, \lambda>0.$$
A direct application of the Schur test shows that $S_b(\iu \sqrt{\lambda})$ is bounded up to a constant by $\sqrt{\lambda}\int_{\R^2}|F(\sqrt{\lambda}|x|)|\d x$, which decays like $\lambda^{-1/2}$ when $\lambda$ grows. Also, since $|{A}(x-x')|\leq |x-x'|$, we have 
$$\Vert T_b(x,x';\sqrt{\lambda})\Vert_{\C^2}\leq b\; \sqrt{\lambda}\; |x-x'|\;  F(\sqrt{\lambda}|x-x'|),\quad \forall \lambda>0.$$
Another Schur estimates suffices to conclude the proof.
\end{proof}

\begin{lemma}\label{lem:restrictionSb}
We have $S_b(\iu \sqrt{\lambda})\varphi \in \mathscr{M}$ for all $\varphi\in C^\infty_0(E,\mathbb{C}^2)$.
\end{lemma}
\begin{proof}
 We know that the right-hand side of \eqref{eq:integralkernelresolventbulk_nomagne} satisfies the boundary condition. The complication is that we now have to multiply the free kernel with the diagonal factor $$\e^{\iu b\Mphase(x,x')}=\sum_{k\geq 0}\frac{\big (\iu b(x_1'-x_1)x_2'\big )^k}{k!}.$$
For a fixed $x_1$, this series is uniformly convergent on the support of $\varphi$. By expanding the above polynomials and coupling all the components of $x'$ with the compactly supported $\varphi$, we see that each term  satisfies the boundary condition, hence the whole series does the same. Moreover, the partial derivatives with respect to $x$ of $\e^{ib\Mphase(x,x')}$ grow polynomially in $x'$, but since $\varphi$ is compactly supported, it does not affect the exponential localization in $|x|$, for more technical details see the anaologous proof in the Schr\"odinger case in \cite{MoscolariStottrup}.
\end{proof}

\begin{proposition}\label{prop:relationofSbandTb}
$S_b(\iu \sqrt{\lambda})$ is an almost resolvent for $H_b^E$, in the sense that 
\begin{align*}
    (H_b^E-\iu \sqrt{\lambda})S_b(i\sqrt{\lambda})f = f + T_b(\iu \sqrt{\lambda})f,
\end{align*}
for every $f \in L^2(E,\mathbb{C}^2)$. Moreover, $H_b^E$ is self-adjoint.  
\end{proposition}
\begin{proof}
According to \eqref{hc1}, the operator $H_b^E$ is the closure of the symmetric operator $\tilde{H}_b$ defined on $\mathscr{M}$.  Fix $f \in L^2(E,\mathbb{C}^2)$ and assume that $(f_n) \subseteq C^\infty_0(E,\mathbb{C}^2)$ satisfies $f_n \to f$ as $n \to \infty$. Then it is enough to prove 
\begin{align*}
    (\tilde{H}_b-i\sqrt{\lambda})S_bf_n=f_n+T_bf_n,
\end{align*}
because the convergence of $S_bf_n$ and $T_bf_n$ follow from Lemma  \ref{prop:boundednessSbTb}, and thus $S_bf$ lies in the domain of $(H_b^E-\iu \sqrt{\lambda})$ by the definition of the closure.

Now let $\psi \in C^\infty_0(E,\mathbb{C}^2)$. Using \eqref{eq:HbEalmostcommutingwithmagneticphase}, \eqref{hc10} and the fact that $\Mphase(x,x)=0$, we have 

\begin{align*}
    ((\tilde{H}_b-\iu \sqrt{\lambda})S_b\psi)(x)&= \int_{E}\d x' \e^{\iu b\Mphase(x,x')}\big ((-\iu \nabla_x\cdot\sigma-\iu \sqrt{\lambda})-b{A}(x-x') \cdot \sigma\big )\K_0^E(x,x')\psi(x') \\
    &=\psi(x)-b\int_{E}\d x' \e^{\iu b\Mphase(x,x')}{A}(x-x') \cdot \sigma\; \K_0^E(x,x')\psi(x')\\
    &=\psi(x)+\big (T_b(\iu \sqrt{\lambda})\psi\big )(x).
\end{align*}
Regarding the self-adjointness of $H_b^E$, we only lack proving that there exists some $\lambda>0$ such that $H_b^E \pm \iu \sqrt{\lambda}$ is surjective. We do this only for $H_b^E - \iu \sqrt{\lambda}$, since the other case is similar. From Lemma \ref{prop:boundednessSbTb} we know that if $0\leq b\leq b_0$, there exists a $\lambda_0>0$ such that $\Vert T_b(i\sqrt{\lambda})\Vert\leq 1/2$ for all $\lambda\geq \lambda_0$. In particular, $\big (1+T_b(\iu \sqrt{\lambda_0})\big )^{-1}$ exists and is bounded, thus 
$$(H_b^E-\iu \sqrt{\lambda_0})S_b(\iu \sqrt{\lambda_0})(1+T_b(\iu \sqrt{\lambda_0})\big )^{-1}f=f,\quad \forall \, f\in L^2(E,\C^2),
$$
which concludes the proof.
\end{proof}

Finally, the proof of Proposition \ref{prop:EssSelfAdj}(i) is just a corollary of Lemma \ref{lem:restrictionSb} and Proposition \ref{prop:relationofSbandTb}.

\subsection{Proof of Proposition \ref{prop:EssSelfAdj}(ii)} 

Consider the edge Dirac-Landau operator $H^E_b$. If $\lambda>\lambda_0$ we can write 
$$(H_b^E-\iu \sqrt{\lambda})^{-1}=S_b(\iu \sqrt{\lambda})(1+T_b(\iu \sqrt{\lambda})\big )^{-1}=\sum_{k\geq 0}S_b(\iu \sqrt{\lambda})T_b^k(\iu \sqrt{\lambda}).$$

From \eqref{eq:Slambda} we have that $S_b(\iu \sqrt{\lambda})T_b^k(\iu \sqrt{\lambda})$ is an integral operator with a kernel given by 
\begin{align*}
b^k\int_{E^k} \d x^{(1)} ... \d x^{(k-1)} & \e^{\iu b(\Mphase(x,x^{(1)})+...+\Mphase(x^{(k-1)},x'))}\\
&\times \K_0(x,x^{(1)};\sqrt{\lambda})...\sigma\cdot {A}(x^{(k-1)}-x')\K_0(x^{(k-1)},x';\sqrt{\lambda}).
\end{align*}

Moreover, we have the phase ``composition identity":
\begin{equation}\label{hc13}
    \Mphase(x,y)+\Mphase(y,x')=\Mphase(x,x')+(y_1-x_1)(y_2-x'_2),
\end{equation}
which leads to 
\begin{align}
\label{eq:phaseCR}
&\Mphase(x,x^{(1)})+...+\Mphase(x^{(k-1)},x')=\Mphase(x,x')\\
&+(x_1^{(1)}-x_1)(x_2^{(1)}-x_2^{(2)})+...+(x_1^{(k-1)}-x_1)(x_2^{(k-1)}-x_2').
\end{align}

By taking out the common factor $\e^{\iu b\Mphase(x,x')}$, the remaining kernel defines an operator which is smooth as a function of $b$ in the operator norm topology. Let us first show that it is differentiable. When we differentiate the factor $b^k$, an application of the Schur test show that the norm of the derivative can be bounded by $kb_0^{k-1} C^k \lambda_0^{-k}$. 

When we differentiate the phase instead, we get $k$ terms, see \eqref{eq:phaseCR}. Consider for example the one containing $(x_1^{(j)}-x_1)(x_2^{(j-1)}-x_2^{(j)})$, for some $1\leq j\leq k$. This term can be rewritten as 
\begin{equation}
\label{eq:rewriting1}
(x_1^{(j)}-x_1)(x_2^{(j-1)}-x_2^{(j)})=\sum_{s=1}^{j}(x_1^{(s)}-x_1^{(s-1)})(x_2^{(j-1)}-x_2^{(j)})
\end{equation}
which are $1\leq j\leq k$ terms containing differences of coordinates from ``consecutive" vectors. The advantage of \eqref{eq:rewriting1} is that the growth of any of the term in the sum is dominated by the exponential decay of one of the integral kernels. By this procedure, we get a number of iterated integrals (number growing like $k^2$) in which at most two kernels from the multiple integrals will be affected by two linearly growing factors. Again by using a Schur test, all these factors can be bounded by a constant of the form
$$\tilde{C} k^2 b^k C^{k-2}\lambda_0^{-k+2},\quad k\geq 2.$$
This shows that the whole series is differentiable on $(0,b_0)$. For higher derivatives, the only difference is that the number of iterated integrals will have a higher polynomial growth in $k$, but this will not affect the radius of convergence of the series (i.e. $b_0$). A similar discussion can be found in \cite{MoscolariStottrup}. 

\qed

\section{Exponential decay, local singularities, and Schatten class properties}\label{sec4}
We will only formulate results for the edge operator, but they also hold true for the bulk operator. 

If $z\in \C$, we write $z=z_1+iz_2$, with $z_1,z_2\in\R$. 
\begin{lemma}\label{lemmahc3} 
 Given $b_0>0$, there exist two positive constants $C,r>0$ such that for any $0\leq b\leq b_0$ and any $z\in \C$ with $0<|z_2|\leq 1$ we have that the  resolvent $(H_b^E-z)^{-1}$ has an integral kernel of potential type obeying: \begin{align*}
     |\K_b^E(x,x';z)\leq C\; (1+|z_1|)^6 |z_2|^{-1} |x-x'|^{-1} \e^{-r\; |z_2|\; |x-x'|},\quad \forall \, x,x'\in E.
 \end{align*}
\end{lemma}
\begin{proof}
We have already constructed the integral kernel of $(H_b^E-i\sqrt{\lambda})^{-1}$ when $\lambda>0$ is large enough. In order to simplify notation, let us put $\lambda=1$. Here we have to extend the construction to any $z$ with a non-zero imaginary part.  
Applying the first resolvent identity six times we can write: 
\begin{equation}\label{hc22}
(H_b^E-z)^{-1}=\sum_{j=1}^6 (z-i)^{j-1}(H_b^E-\iu )^{-j} +(z-\iu )^6(H_b^E-\iu )^{-3}(H_b^E-z)^{-1}(H_b^E-\iu )^{-3}.
\end{equation}
The operator $(H_b^E-\iu )^{-2}$ has an integral kernel with a milder singularity on the diagonal (only logarithmic, \cite{V, MoscolariStottrup}), and there exist two constants $C_1,C_2 >0$ such that its integral kernel (after taking the $\C^2$ norm) is pointwise bounded by 
$$\Vert(H_b^E-\iu )^{-2}(x,x')\Vert_{\C^2}\leq  C_1 (1+ |\log(|x-x'|)|)\e^{-C_2|x-x'|}.$$
For powers larger than $2$ of the resolvent, the corresponding kernels do not have  singularities anymore and are jointly continuous. 

Using the identity $$ \e^{\alpha \langle x-x_0\rangle} (H_b^E-z)\e^{-\alpha \langle x-x_0\rangle}=\Big (1 +\iu \alpha (\nabla_x \langle x-x_0\rangle)\cdot\sigma\;  (H_b^E-z)^{-1}\Big )(H_b^E-z)$$
and the fact that 
$$1 +\iu \alpha (\nabla_x \langle x-x_0\rangle)\cdot\sigma\;  (H_b^E-z)^{-1}$$
is invertible when $\alpha$ is of order of $|z_2|$, uniformly in $x_0$, a standard Combes-Thomas argument shows that there exist two  constants $C_3>0$ and $0<r<1$ such that 
$$\Vert \e^{r|z_2| \langle\cdot -x_0\rangle }(H_b^E-z)^{-1}\e^{- r|z_2| \langle \cdot -x_0\rangle }\Vert \leq C_3\; |z_2|^{-1},\quad \forall \, 0<|z_2|\leq 1,\; \forall \, x_0\in E.$$
There exists some $\alpha_0>0$ such that $\Vert(H_b^E-\iu )^{-3}(x,\cdot)\Vert_{\C^2}\e^{\alpha_0|\cdot -x|}$ is in $L^2(E)$. Hence the operator $(H_b^E-\iu)^{-3}(H_b^E-z)^{-1}(H_b^E-\iu )^{-3}$ has an integral kernel given by 
\begin{align*}
&K_7(x,y;z):=(H_b^E-\iu)^{-3}(x,\cdot )(H_b^E-z)^{-1}(H_b^E-\iu )^{-3}(\cdot,y)\\
&=(H_b^E-\iu )^{-3}(x,\cdot )\e^{-r|z_2| \langle \cdot-y\rangle} \Big (\e^{r|z_2| \langle \cdot-y\rangle}(H_b^E-z)^{-1}\e^{-r|z_2| \langle \cdot-y\rangle} \Big )\e^{r|z_2| \langle \cdot-y\rangle}(H_b^E-\iu )^{-3}(\cdot,y).
\end{align*}
This shows that $K_7(x,y;z)$ is bounded for all $x,y$ and has the claimed exponential localization around the diagonal: 
\begin{align}\label{hc21}
  |K_7(x,y;z)|\leq C |z_2|^{-1} \e^{-r|z_2|\; |x-y|}.  
\end{align}

\end{proof}

An important corollary of Lemma \ref{lemmahc3} is the following trace class property.

\begin{lemma}\label{lemmahc4}
Let $\chi_M$ be the indicator function of a bounded set $M\subset E$. Then the operator $\chi_M (H_b^E-i)^{-4}$ is trace class.
\end{lemma}
\begin{proof}
We can write 
$$\chi_M (H_b^E-\iu)^{-4}=\chi_M \e^{\alpha\langle \cdot\rangle} \Big ( \e^{-\alpha\langle \cdot\rangle/2}\e^{-\alpha\langle \cdot\rangle/2} (H_b^E-\iu)^{-2}\e^{\alpha\langle \cdot\rangle/2}\Big )\e^{-\alpha\langle \cdot\rangle/2}(H_b^E-\iu )^{-2},$$
where the factor $\chi_M \e^{\alpha\langle \cdot\rangle}$ is bounded, while the other two factors can be shown to be Hilbert-Schmidt operators by using the estimate proved in  Lemma \ref{lemmahc3}.
\end{proof}

We now introduce the Helffer-Sjöstrand formula \cite{HS}. Let $f$ be a real valued Schwartz function on $\mathbb{R}$.  Define for $N \geq 1$ the almost analytic continuation $f_N$ given by 
\begin{align*}
   f_N(z_1+\iu z_2)\coloneqq g(z_2)\sum_{j=0}^N \frac{1}{j!}\frac{\partial^j f}{\partial z_1^j}(z_1)(iz_2)^j,
\end{align*}
where $\Re(z)=z_1$, $\Im(z)=z_2$ and $g \in C^\infty_c(\mathbb{R})$ satisfying $0\leq g \leq 1$, $g(x) = 1$ when $\vert x \vert \leq 1/2$ and $g(x)=0$ when $\vert x \vert > 1$. Then $f_N$ has the following properties:
\begin{enumerate}
    \item $f_N(x)=f(x)$, for all $x \in \mathbb{R}$.
    \item $\supp(f_N) \subseteq \mathcal{D} \coloneqq  \{ z \in \mathbb{C} \mid \vert \Im(z) \leq 1\}$.
    \item There exists $C_N>0$ such that
    \begin{align*}
         \vert \Bar{\partial}f_N(z)\vert 
         \leq C_N \frac{\vert z_2 \vert^N}{\langle z_1 \rangle^N},
    \end{align*}
    where $\Bar{\partial} = \partial_{z_1} + \iu \partial_{z_2}$ and $\langle x \rangle = (1+\vert x \vert^2)^{1/2}$ and we have used that $f$ and all its derivatives decay faster than any polynomial.
\end{enumerate}
Then by the Helffer-Sjöstrand formula we can write $f(H_b^{\star})$, where $\star \in \{\emptyset,E\}$ as
\begin{align}\label{dc1}
    f(H_b^\star) = -\frac{1}{\pi}\int_\mathcal{D} \Bar{\partial}f_N(z)(H_b^\star-z)^{-1} \d z_1 \d z_2,
\end{align}
where $\d z_1 \d z_2$ denotes the Lebesgue measure of $\mathbb{C} \cong \mathbb{R} \times \mathbb{R}$. The properties of $f(H_b^{\star})$ are summarized in the next lemma.

\begin{lemma}\label{lemmahc5}
 Let $f$ be a Schwartz function. Then the operators $f(H_b)$ and $f(H_b^E)$ have jointly continuous integral kernels and for every $M>0$ there exists $C_M>0$ such that
\begin{align*}
    \vert f(H_b^\star)(x,x')\vert \leq C_M \langle x-x'\rangle^{-M},
\end{align*}
for $\star \in \{\emptyset,E\}$. 

Furthermore, let $\chi_L$ be the indicator function over the strip $[0,1]\times[0,L]$, for $L\geq 1$. Then $\chi_L f(H_b^\star)$ is trace class and 
\begin{align*}
    \Tr (\chi_L f(H_b^\star))=\int_{\mathbb{R}^2} \chi_L(x)f(H_b^\star)(x,x)\d x,
\end{align*}
for $\star \in \{\emptyset,E\}$.
\end{lemma}
\begin{proof}
The fact that $f(H_b^\star)$ is locally trace class is a consequence of Lemma \ref{lemmahc4} together with the fact that the function $(t-i)^4f(t)$ is bounded. Then, the existence of jointly continuous integral kernels can  be shown by mimicking the proof of \cite[Lemma A.2]{CorneanMoscolariTeufel}. 
Then, by \eqref{hc22} in the formula of $f(H_b^\star)$, where the analytic part vanishes under the complex integral, we get
\begin{equation}
\label{eq:aux1}
f(H_b^\star)(x,x')=-\pi^{-1}\int_\mathcal{D} \Bar{\partial}f_N(z)(z-\iu)^{6} K_7(x,x';z) \d z_1 \d z_2.
\end{equation}
If we multiply \eqref{eq:aux1} with $\langle x-x'\rangle^M$, we see that due to \eqref{hc21} we can bound the polynomial growth at the expense of a factor $|z_2|^{-M}$, but which can also be controlled by choosing an almost analytic extension with a large enough $N$. 
\end{proof}

\section{Proof of Theorem \ref{thm:DiracBulkEdge}}\label{sec5}

We follow the same ideas and steps as in \cite{CorneanMoscolariTeufel}, with the remark that in the Dirac case, the ``edge current operator"  $\iu [H_b^E,X_1]$ is of order zero and equals $\sigma_1$. 

\subsection{Geometric perturbation theory for Dirac operators}
Before proving (i) we need to show that the resolvent of the edge operator $H_b^E$ and the resolvent of the bulk operator $H_b$ are "close" to each other far away in the bulk. Let $L \leq 1$ be fixed and define
\begin{align*}
\Xi_L(t) \coloneqq \{x \in E \mid \dist(x,\partial E) \leq t \sqrt{L}\},
\end{align*}
where $t > 0$. Furthermore, let both $0 \leq \eta_0,\eta_L \leq 1$ be smooth and non-negative functions independent of $x_1$ satisfying that $\eta_0(x)+\eta_L(x)=1$, for every $x \in E$ along with the assumptions:
\begin{enumerate}
    \item $\supp (\eta_0) \subseteq \Xi_L(2) $,
    \item $\supp (\eta_L) \subseteq E \setminus \Xi_L(1)$,
    \item $\Vert \partial_{x_2}^n \eta_k \Vert_\infty \leq C_n L^{-n/2}$, for $n \in \mathbb{N}$, $k \in \{0,L\}$ and $C_n>0$.
\end{enumerate}
We now introduce two similar functions, which can be considered as a stretched out versions of $\eta_0$ and $\eta_L$. Let both $0 \leq \tilde{\eta}_0,\tilde{\eta}_L \leq 1$ be smooth and non-negative functions independent of $x_1$ satisfying:
\begin{enumerate}[label=(\Roman*)]
    \item $\supp (\tilde{\eta}_0) \subseteq \Xi_L(11/4) $,
    \item $\supp (\tilde{\eta}_L) \subseteq E \setminus \Xi_L(1/4)$,
    \item $\tilde{\eta}_k\eta_k = \eta_k$, for $k \in \{0,L\}$,
    \item $\dist(\supp(\partial_{x_2} \tilde{\eta}_k), \supp(\eta_k)) \sim C\sqrt{L}$, for $k \in \{0,L\}$ and $C > 0$,
    \item $\Vert \partial_{x_2}^n \tilde{\eta}_k \Vert_\infty \leq C_n L^{-n/2}$, for $n \in \mathbb{N}$, $k \in \{0,L\}$ and $C_n > 0$.
\end{enumerate}
Note that $\eta_k$ and the derivative of $\tilde{\eta}_k$ have disjoint supports.

\begin{proposition}\label{prop:HbedgeequalHbonHbdomain}
The multiplication by $\tilde{\eta}_L$ maps $ \mathscr{D}(H_b)$ into $\mathscr{D}(H_b^E)$ and
\begin{align*}
    H_b^E \tilde{\eta}_L \psi = H_b \tilde{\eta}_L \psi,
\end{align*}
for every $\psi \in \mathscr{D}(H_b)$.
\end{proposition}
\begin{proof}
Let $\psi \in \mathscr{D}(H_b)$ and $(\psi_n) \subseteq C^\infty_0(\mathbb{R}^2,\mathbb{C}^2)$ be a sequence for which
\begin{align*}
    \psi_n \to \psi \qquad \textup{ and } \qquad \tilde{H}_b \psi_n \to H_b \psi
\end{align*}
in $L^2(\mathbb{R}^2,\mathbb{C}^2)$, when $n \to \infty$. Such a sequence exists since $H_b$ is essentially self-adjoint on  $C^\infty_0(\mathbb{R}^2,\mathbb{C}^2)$. By the properties of $\tilde{\eta}_L$ it follows that $(\tilde{\eta}_L\psi_n) \subset C^\infty_0(E,\mathbb{C}^2)$ and when $n \to \infty$ we have $\tilde{\eta}_L\psi_n \to  \tilde{\eta}_L\psi$ and:
\begin{align*}
\tilde{H}_b(\tilde{\eta}_L\psi_n) =(-\iu \nabla \tilde{\eta}_L)\cdot \sigma \psi_n +\tilde{\eta}_L\tilde{H}_b\psi_n\to (-\iu \nabla \tilde{\eta}_L)\cdot \sigma \psi +\tilde{\eta}_LH_b\psi=H_b \tilde{\eta}_L \psi,
\end{align*}
 which shows that $\tilde{H}_b(\tilde{\eta}_L\psi_n)$ converges to $H_b \tilde{\eta}_L \psi$, and hence by the definition of the closure
\begin{align*}
    H_b^E(\tilde{\eta}_L\psi) =\lim_{n\to\infty}\tilde{H}_b(\tilde{\eta}_L\psi_n)= \tilde{H}_b(\tilde{\eta}_L\psi).
\end{align*}
\end{proof}

Let $z\in \C$ with a non-zero imaginary part. We  define in $L^2(E,\mathbb{C}^2)$ the bounded operator
\begin{align}\label{hc16}
    U_L(z) \coloneqq \tilde{\eta}_L(H_b-z)^{-1}\eta_L + \tilde{\eta}_0(H^E_b - z)^{-1}\eta_0.
\end{align}
Then by the properties of $\eta_k$ and $\tilde{\eta_k}$, for $k \in \{0,L\}$ along with Proposition~\ref{prop:HbedgeequalHbonHbdomain} we have 
\begin{align*}
    (H_b^E-z)U_L(z) 
    &= 1 + W_L(z),
\end{align*}
 $W_L$ is the bounded operator given by
\begin{align*}
    W_L(z) \coloneqq -\iu (\nabla \tilde{\eta}_L) \cdot \sigma\; (H_b - z)^{-1}\eta_L  -\iu (\nabla \tilde{\eta}_0) \cdot \sigma\; (H_b^E - z)^{-1}\eta_0.
\end{align*}
This leads to the following identity satisfied by the resolvent of the edge operator
\begin{align*}
    (H^E_b-z)^{-1} = U_L(z) -(H^E_b-z)^{-1} W_L(z).
\end{align*}
Using the Helffer-Sj\"ostrand formula this leads to: 
\begin{align}\label{hc17}
    f(H^E_b) = \tilde{\eta}_L f(H_b)\eta_L +\tilde{\eta}_0 f(H_b^E)\eta_0  +\pi^{-1}\int_{\mathcal{D}}\overline{\partial}f_N(z)(H^E_b-z)^{-1} W_L(z) \d z_1 \d z_2.
\end{align}


\subsection{Proof of Theorem \ref{thm:DiracBulkEdge} \ref{thm:DiracBulkEdge1}}

The proof of the first limit in \eqref{hc3'} is more or less identical with the Schr\"odinger case \cite{CorneanMoscolariTeufel} and it is based on \eqref{hc17}. Namely, we see that the trace of $L^{-1}\chi_L\tilde{\eta}_0 f(H_b^E)\eta_0 $ is bounded by $1/\sqrt{L}$, while the trace of $L^{-1}\chi_L \tilde{\eta}_L f(H_b)\eta_L$ converges to ${\rm Tr}(\chi_\Omega f(H_b))$ due to the fact that the diagonal of the integral kernel of $f(H_b)$ is $\Z^2$-periodic.  Now let us prove that when the last term in \eqref{hc17} is multiplied with $\chi_L$, its trace decays faster than any power of $L$. Indeed, by using a similar trick as in \eqref{hc22}, the terms with singular kernels are analytic hence do not contribute to the $z$ integral. Also, because the distance between the support of $\nabla\eta_0$ and the support of $\eta_0$ is of order $\sqrt{L}$ (and the same is true for the other cut-off), we obtain a decaying factor which goes like $\e^{-c|z_2|\sqrt{L}}$ for every $z\in \mathcal{D}$. This factor decays faster than any power of $L$ at the expense of an equally high power in $|z_2|^{-1}$, which can be controlled by the choice of $N$ in the construction of $f_N$. 

The proof of the second limit, namely $\lim_{L\to\infty}\rho_L'(b)=B_f'(b)$, is more involved and needs magnetic perturbation theory adapted to Dirac operators. Let us fix $z$ with a non-zero imaginary part and $b_0\geq 0$. For $\epsilon>0$ we define the operators  $S_\epsilon(z)$ and $T_\epsilon(z)$ through their integral kernels 
$$S_\epsilon(x,x';z):=\e^{i\epsilon\Mphase(x,x')} \K_{b_0}^E(x,x';z)\; {\rm and}\; T_\epsilon(x,x';z):=-\epsilon \e^{i\epsilon\Mphase(x,x')}\sigma \cdot{A}(x-x') \K_{b_0}^E(x,x';z).$$

Then we have 
$(H_{b_0+\epsilon}^E-z)S_\epsilon(z)=1+T_\epsilon(z)$  or by iterating:

\begin{equation}\label{hc20}
(H_{b_0+\epsilon}^E-z)^{-1}=
    S_\epsilon(z)-S_\epsilon(z)
    T_\epsilon(z)+(H_{b_0+\epsilon}^E-z)^{-1} T^2_\epsilon.
\end{equation}
Using the composition rule \eqref{hc13} we have that the integral kernel of $S_\epsilon T_\epsilon$ is given by:
\begin{align*}
(S_\epsilon T_\epsilon)(&x,x';z)=-\epsilon \e^{i\epsilon\Mphase(x,x')}\int_E \d y\; \K_{b_0}^E(x,y;z)\sigma \cdot{A}(y-x') \K_{b_0}^E(y,x';z)\\
&-\epsilon e^{i\epsilon\Mphase(x,x')}\int_E \d y\;\Big (e^{i\epsilon(y_1-x_1)(y_2-x_2')}-1\Big ) \K_{b_0}^E(x,y;z)\sigma \cdot{A}(y-x') \K_{b_0}^E(y,x';z).
\end{align*}
This integral kernel is already jointly continuous in $x$ and $x'$ and moreover, using the bound 
$$\big \vert \e^{i\epsilon(y_1-x_1)(y_2-x_2')}-1\big \vert \leq \epsilon \; |x-y|\; |y-x'|$$
together with Lemma \ref{lemmahc3}, there exist three constants $C,p,q>0$ such that 
\begin{align}\label{hc23}
\Big |(S_\epsilon T_\epsilon)(&x,x;z)+\epsilon \int_E \d y\; \K_{b_0}^E(x,y;z)\sigma \cdot{A}(y-x) \K_{b_0}^E(y,x;z)\Big |\leq C(1+|z_1|)^p|z_2|^{-q} \epsilon^2.
\end{align}
The last term on the right-hand side of  \eqref{hc20} has an integral kernel which can be bounded by  a similar expression as the one in the right-hand side of \eqref{hc23}, since $T_\epsilon^2$ is already of second order in $\epsilon$. 

Moreover, by using the Helffer-Sj\"ostrand formula in \eqref{hc20} we obtain 
\begin{align}\label{hc24}
f(H_{b_0+\epsilon}^E)(x,x)&=f(H_{b_0}^E)(x,x)\\
&-\epsilon \pi^{-1} \int_\mathcal{D} \Bar{\partial}f_N(z)\int_E \d y\; \K_{b_0}^E(x,y;z)\sigma \cdot{A}(y-x) \K_{b_0}^E(y,x;z)\d z_1 \d z_2 +\mathcal{O}(\epsilon^2),\nonumber
\end{align}
uniformly in $x$. In particular, using the definition of $\rho_L(b)$ from \eqref{hc3}, the estimate \ref{hc24} leads to
\begin{equation}\label{hc25}
\rho_L'(b_0)=-\pi^{-1} L^{-1}\int_{S_L}\d x\int_E \d y\int_\mathcal{D}\d z_1 \d z_2\; \Bar{\partial}f_N(z)\;  \K_{b_0}^E(x,y;z)\sigma \cdot{A}(y-x) \K_{b_0}^E(y,x;z).
\end{equation}
Also, \eqref{hc24} leads to 
$$\rho_L'(b_0)=\frac{\rho_L(b_0+\epsilon)-\rho_L(b_0)}{\epsilon}+\mathcal{O}(\epsilon)$$
where a crucial remark is that the remainder is uniformly bounded in $L$. Since $\rho_L$ itself converges to $B_f$,  this implies that $\limsup\rho_L'(b_0)$ and $\liminf\rho_L'(b_0)$ can at most differ by something of order $\epsilon$, thus they must be equal to each other. Now by taking the limit we have 
$$\lim_{L\to\infty} \rho_L'(b_0)=\frac{B_f(b_0+\epsilon)-B_f(b_0)}{\epsilon}+\mathcal{O}(\epsilon).$$
Thus $B_f'(b_0)$ exists and $\lim_{L\to\infty} \rho_L'(b_0)=B_f'(b_0)$. In any case, one can also prove the existence of $B_f'(b_0)$ directly from the definition of $B_f$. 

\subsection{Proof of Theorem \ref{thm:DiracBulkEdge}(ii)}

The remaining task is to take the limit $L\to\infty$ in \eqref{hc25} in a clever way. Again the proof is similar to the Schr\"odinger case \cite{CorneanMoscolariTeufel} and here we summarize the main steps: 
\begin{enumerate}
 \item[1.] We have $\sigma\cdot {A}(y-x)=(x_2-y_2)\sigma_1$. We also have that $\iu [H_b^E,X_1]=\sigma_1$. Hence the part of the integral in \eqref{hc25} which only contains $-x_2\sigma_1$ is proportional up to a constant with 
$${\rm Tr}(\chi_L [f(H_b^E),X_1]X_2)={\rm Tr}( [\chi_L X_2 f(H_b^E),X_1])=0$$
as the trace of a total commutator. Hence only the term $y_2\sigma_1$ counts, but for later reasons we write $y_2-L$ instead of $L$; this does not change the value of the integral because the added term is again proportional with the trace of a  commutator. Therefore we get:

\begin{equation}\label{hc26}
\begin{aligned}
\rho_L'(b_0)=-\frac{1}{\pi L}\int_{S_L}\d x\int_E \d y\int_\mathcal{D}\d z_1 \d z_2\; \Bar{\partial}f_N(z)\;  \K_{b_0}^E(x,y;z)(L-y_2)\sigma_1 \K_{b_0}^E(y,x;z).
\end{aligned}
\end{equation}

\item[2.] By restricting the integral in $y$ to the strip $0<y_2<L$, the error we make is of order $\mathcal{O}(L^{-1})$; this is because when $y_2>L$ we have that $\e^{-r|z_2||y_2-x_2|}=\e^{-r|z_2|(y_2-L)}\e^{-r|z_2|(L-x_2)}$, hence we have a strong localization of the integrand near $x_2=L$ while the integral over $y_2>L$ is convergent. Thus we have:

\begin{align}\label{hc27}
&\lim_{L\to\infty}\rho_L'(b_0)\\
&=-\lim_{L\to\infty}\frac{1}{\pi L}\int_{S_L}\d x\int_{0<y_2<L} \d y\int_\mathcal{D}\d z_1 \d z_2\; \Bar{\partial}f_N(z)\;  \K_{b_0}^E(x,y;z)(L-y_2)\sigma_1 \K_{b_0}^E(y,x;z).\nonumber
\end{align}

\item[3.]  By extending the integral over $x$ to the infinite strip $S_\infty$, the error we make is again of order $L^{-1}$; this is so because if $x_2>L$ and $y_2<L$ we may again write $\e^{-r|z_2||y_2-x_2|}=\e^{-r|z_2|(L-y_2)}\e^{-r|z_2|(x_2-L)}$ and both spatial integrals are bounded in $L$. Hence we obtain:
\begin{align}\label{hc28}
&\lim_{L\to\infty}\rho_L'(b_0)\\
&=-\lim_{L\to\infty}\frac{1}{\pi L}\int_{[0,1]\times \R}\d x\int_{0<y_2<L} \d y\int_\mathcal{D}\d z_1 \d z_2\; \Bar{\partial}f_N(z)\;  \K_{b_0}^E(x,y;z)(L-y_2)\sigma_1 \K_{b_0}^E(y,x;z).\nonumber
\end{align}

\item[4.] At this stage we use the covariance property of the integral kernels, which is a consequence of the fact that $H_b^E$ is $\Z$-periodic in the $x_1$ variable and commutes with translations parallel to the axis defined by $x_2=0$. This implies that 
$$\K_{b_0}^E(x_1,x_2;y_1+\gamma_1,y_2;z)=\K_{b_0}^E(x_1-\gamma_1,x_2;y_1,y_2;z),\quad \forall \gamma_1\in \Z.$$
By writing 
$$\int_{0<y_2<L} F(y)\d y=\sum_{\gamma_1\in \Z}\int_0^1 \d y_1\int_0^L\d y_2 F(y_1+\gamma_1,y_2)$$
we see that the spatial integrals in \eqref{hc28} have the property: 
$$\int_0^1 \d x_1 \int_0^\infty \d x_2 \int_{\R}\d y_1\int_{0}^L \d y_2=\int_E \d x \int_0^1\d y_1\int_{0}^L \d y_2=\int_E \d x \int_{S_L}\d y.$$
The integral with respect to $x$ over the whole $E$ generates a resolvent raised to power $2$. Furthermore, the integral over $z$ generates the operator $-f'(H_{b_0}^E)$, since $$f'(t)=\pi^{-1}\int_\mathcal{D}\d z_1 \d z_2\; \Bar{\partial}f_N(z)(t-z)^{-2}.$$ Thus:
\begin{align}\label{hc29}
&B_f'(b_0)=\lim_{L\to\infty}\rho_L'(b_0)=-\lim_{L\to\infty}\int_0^1\d y_1\int_0^L\d y_2   (1-y_2/L)\sigma_1 f'(H_{b_0}^E)(y,y),
\end{align}
which proves Theorem \ref{thm:DiracBulkEdge}(ii) in the particular case when $g(t)=1-t$ if $0\leq t\leq 1$. The extension to all smooth $g$ with $g(0)=1$ and $g(1)=0$ can be done by mimicking the method in \cite{CorneanMoscolariTeufel}. 
\end{enumerate}

\subsection{Proof of Theorem \ref{thm:DiracBulkEdge}(iii)}

Now let us sketch an argument for why \eqref{hc2} holds. When $f'$ is supported in a spectral gap of the bulk operator $H_b$, then $f'(H_b)=0$. By using \eqref{hc17} with $f$ replaced by $f'$ and with $L=1$, we observe that the two non-zero terms contain factors which are compactly supported in the $x_2$ variable. Since $\chi_\infty$ localizes in $0\leq x_1\leq 1$, it turns out that the operator $\chi_\infty f'(H_b^E) $ is trace class. This is proved by showing that it can be written as a finite sum of products involving two Hilbert-Schmidt operators. Therefore the function $\chi_\infty(y){\rm tr}_{\C^2}\big (\sigma_1 f'(H_b^E)(y,y)\big )$ is in $L^1(E)$ and its integral equals the trace of $\chi_\infty \sigma_1 f'(H_b^E)$. Eventually \eqref{hc2} follows from \eqref{hc29} after an application of the Lebesgue dominated convergence theorem. 

\subsection{Proof of Theorem \ref{thm:DiracBulkEdge} \ref{thm:DiracBulkEdge4}}
Following \cite{BS1}, the purely magnetic Dirac-Landau edge operator commutes with all the translations in the $x_1$ direction and can be written as a fiber integral of an one-dimensional Dirac operator $h_b^E(\xi_1)=-\iu \frac{d}{dx_2}\sigma_2+(bx_2+\xi_1)\sigma_1$ living in $L^2([0,\infty))\otimes \C^2$, with the boundary condition $\sigma_1\Psi(0)=\Psi(0)$. The fiber operator has compact resolvent and purely discrete spectrum. There exist some positive non-degenerate fiber eigenvalues $\lambda_k(\xi_1)$, $k\geq 0$, $\xi_1\in \R$, which are strictly increasing, and  $$\lim_{\xi_1\to-\infty}\lambda_k(\xi_1)=\sqrt{2kb},\quad  \lim_{\xi_1\to\infty}\lambda_k(\xi_1)=\infty,\quad k\geq 0.$$ 
Moreover, there exist some negative non-degenerate fiber eigenvalues $\lambda_k(\xi_1)$, $k\leq -1$, $\xi_1\in \R$, which satisfy $$\lim_{\xi_1\to-\infty}\lambda_k(\xi_1)=-\sqrt{2|k|b},\quad \lim_{\xi_1\to\infty}\lambda_k(\xi_1)=-\infty,\quad k\leq -1.$$ 
The negative eigenvalues have a negative global maximum. In particular, the purely magnetic edge operator has a spectral gap  $(-\sqrt{b}\; \theta,0)$ where $0<\theta<\sqrt{2}$. More about this can be found in Theorem 4.3 in \cite{BS1}. 

 Now let us assume that $f'$ has a support which does not include any of the bulk Dirac-Landau levels $\pm \sqrt{2nb}$, $n\geq 0$. Thus \eqref{hc2} holds true. Remember also that  $\sigma_1= \iu [H_b^E,X_1]$. Using the identity 
 $$f'(H_b^E)(x,y)=\frac{1}{2\pi}\int_\R\d \xi_1 \; \e^{\iu \xi_1(x_1-y_1)}f'(h_b^E(\xi_1))(x_2,y_2),$$
 we have 
 \begin{align*}
 -{\rm Tr}\big (\chi_\infty f'(H_b^E)\sigma_1\big )&=-\frac{1}{2\pi}\int_\R\d \xi_1 \; {\rm Tr}\big (f'(h_b^E(\xi_1))\partial_{\xi_1}h_b^E(\xi_1)\;  \big )\\
 &=\frac{1}{2\pi}\sum_{k\in\Z }\big (f(\lambda_k(-\infty))-f(\lambda_k(+\infty))\big )=\frac{1}{2\pi}\sum_{k\in\Z }f\big ({\rm sgn}(k)\sqrt{2|k|b}\big ),
 \end{align*}
where in the second equality we used Feynman-Hellmann and the fact that $f'$ is not supported near the bulk energy (i.e. $f$ is constant in a neighborhood of every such energy level). Now if we take an $f$ which has compact support and equals $1$ near $N$ different bulk levels, the corresponding $B_f(b)$ will be the integrated density of states of the bulk projection corresponding to these $N$ levels.  Then \eqref{hc2} coupled with \eqref{hc4} implies that the Chern character must be equal to $N$. 

Finally, let us give a few hints for the proof of \eqref{dc2}. This formula can be obtained by computing the integral kernel of the purely magnetic Dirac-Landau resolvent, followed by the one of $f(H_b)$. The main idea is to use the identity
$$\big ((p-b{A})\cdot \sigma -z\big )^{-1}=\big ((p-b{A})\cdot \sigma +z\big )\big ([(p-b{A})^2 -z^2]I_2 +b\sigma_3\big )^{-1}, $$
to identify the diagonal operator-valued entries (the only ones contributing to the local trace): 
$$z\big ((p-b{A})^2 -z^2 \pm b\big )^{-1}=\sum_{n\geq 0}\frac{z}{(2n+1)b\pm b-z^2} \, L_n $$
where $L_n$ are the spectral projections onto the eigenspace associated to the $n$-th eigenvalue of the Landau operator in  $L^2(\R^2)$, followed by an application of the Helffer-Sj\"ostrand formula \eqref{dc1} (or plain residue calculus if $f$ is analytic near the real axis), and finally using that the diagonal value of the integral kernel of each $L_n$ equals $\frac{b}{2\pi}$.

\addtocontents{toc}{\protect\setcounter{tocdepth}{4}}

\section{Proof of Theorem \ref{thm:DiracGapLabel}}\label{sec6}

In this proof it is convenient to consider the transverse gauge $A_t(x):=2^{-1}(-x_2,x_1)$. Let $\Phi(x)=-x_1x_2/2$. Then 
$$\big (-\iu \nabla_x -b{A}_t(x)\big )\cdot\sigma =\e^{ib\Phi(x)}\big (-\iu \nabla_x -b{A}(x)\big )\cdot\sigma\; \e^{-\iu b\Phi(x)}.$$
This shows that quantities like $B_f(b)$ are gauge invariant. From now on, $H_b$ will denote the operator with the transverse gauge.  

As we have already explained in Remark \ref{remarkhc2}, the projection $\Pi_b$ can be written as $f(H_b)$ for some cleverly chosen smooth $f$. Using Lemma \ref{lemmahc5} one can prove that $\Pi_b$ has a jointly continuous integral kernel. However, this can also be done directly: First, write the formula \eqref{hc22} with $H_b$ instead of $H_b^E$ and with $z$ belonging to the contour $\mathcal{C}$. Second, we have
\begin{align*}
    \Pi_b=\frac{\iu}{2\pi} \oint_\mathcal{C} (z-\iu)^6(H_b-\iu)^{-3}(H_b-z)^{-1}(H_b-\iu)^{-3}\d z.
\end{align*}
Reasoning as in the proof of \eqref{hc21}, we see that now we can obtain a similar estimate but with $|z_2|$ replaced by the distance between $\mathcal{C}$ and the spectrum om $H_b$. We conclude that given some $b_0$ there exists some $C,\alpha>0$ such that 
\begin{align}\label{hc30}
    \Vert\Pi_{b_0}(x,x')\Vert_{\C^2}\leq C \; \e^{-\alpha |x-x'|},\quad \forall x,x'\in \R^2.
\end{align}
Now define the anti-symmetric phase 
$\phi_{\ast}(x,x'):=(x_2x'_1-x_1x_2')/2$. This phase has the property that 
$$\big (-\iu\nabla_x-bA_t(x)\big )\cdot \sigma \e^{\iu b\phi_{\ast}(x,x')}=\e^{\iu b\phi_{\ast}(x,x')}\big (-\iu \nabla_x-bA_t(x-x')\big )\cdot \sigma.$$

For $\epsilon\in \R$ and $z\in\mathcal{C}$ we define the bounded operators  $S_\epsilon(z)$ and $T_\epsilon(z)$ through their integral kernels 
$$S_\epsilon(x,x';z):=\e^{\iu \epsilon\phi_{\ast}(x,x')} \K_{b_0}(x,x';z)\; {\rm and}\; T_\epsilon(x,x';z):=-\epsilon \e^{\iu \epsilon\phi_{\ast}(x,x')}\sigma \cdot{A}_t(x-x') \K_{b_0}(x,x';z).$$
Then, we have 
$(H_{b_0+\epsilon}-z)S_\epsilon(z)=1+T_\epsilon(z)$, which shows in particular that if $|\epsilon|$ is small enough, then $z$ is in the resolvent set of $H_{b_0+\epsilon}$  and we have
$$(H_{b_0+\epsilon}-z)^{-1}=
    S_\epsilon(z)\big (1+T_\epsilon(z)\big)^{-1},$$
which implies
\begin{equation}\label{hc31}
(H_{b_0+\epsilon}-z)^{-1}=
    S_\epsilon(z)-(H_{b_0+\epsilon}-z)^{-1} T_\epsilon(z).
\end{equation}
This proves that there exists some $C,\alpha,\epsilon_0>0$ such that for every $| \epsilon|\leq \epsilon_0$ we have
\begin{equation}
\label{eq:aux2}
\big \Vert\Pi_{b_0+\epsilon}(x,x')-\e^{i\epsilon\phi_{\ast}(x,x')} \Pi_{b_0}(x,x')\big \Vert_{\C^2}\leq C\; |\epsilon| \; \e^{-\alpha |x-x'|},\quad \forall \, x,x'\in \R^2.
\end{equation}
The estimate \eqref{eq:aux2} together with \eqref{hc30} are the fundamental ingredients needed in \cite{CorneanMonacoMoscolari} in order to make their method to work. The rest of the proof of Theorem \ref{thm:DiracGapLabel} follows the same steps as in the Schr\"odinger case. 
\qed

\subsection*{Acknowledgments}
HC acknowledges partial support from Danmarks Frie Forskningsfond grant 8021-00084B. The work of MM is supported by a fellowship of the Alexander von Humboldt Foundation. KS is especially grateful to S. Teufel for extending an invitation to visit the University of Tübingen during the spring semester of 2022 and to ``Augustinus Fonden'', ``Knud Højgaards Fond'' and ``William Demant Fonden'' for financially supporting the visit.

\vfill

\end{document}